\documentclass[submission,copyright,creativecommons]{eptcs}
\providecommand{\event}{AMSLO23}

\usepackage{amsmath,amssymb}
\usepackage{amsthm}
\usepackage{proof}
\usepackage{stmaryrd}
\usepackage{latexsym}
\usepackage{url}
\usepackage{mycommands}
\usepackage{color}
\usepackage{times}
\usepackage{tikz}
\usetikzlibrary{positioning}
\usetikzlibrary{matrix}
\usetikzlibrary{calc}
\usetikzlibrary{backgrounds}
\usepackage[all]{xy}
\RequirePackage{txfonts}
\usepackage[show]{ed}
\usepackage{hyperref}
\usepackage{graphicx}
\usepackage{wrapfig}
\usepackage{wasysym}
\usepackage{mathtools}
\usepackage{cleveref}
\usepackage{float}
\usepackage{caption}
\usepackage{subcaption}
\usepackage{mathdots}
\usepackage{comment}
\usepackage{scalerel,stackengine}
\stackMath
\newcommand\reallywidehat[1]{%
\savestack{\tmpbox}{\stretchto{%
  \scaleto{%
    \scalerel*[\widthof{\ensuremath{#1}}]{\kern.1pt\mathchar"0365\kern.1pt}%
    {\rule{0ex}{\textheight}}
  }{\textheight}%
}{2.4ex}}%
\stackon[-6.9pt]{#1}{\tmpbox}%
}
\title{Explorations in Subexponential Non-associative Non-commutative Linear Logic}
\def\titlerunning{Explorations in Subexponential Non-associative Non-commutative Linear Logic}
\author{Eben Blaisdell\institute{University of Pennsylvania, USA}\email{ebenb@sas.upenn.edu}
    \and
    Max Kanovich\institute{University College London, UK}\email{m.kanovich@ucl.ac.uk}
    \and 
    Stepan L. Kuznetsov\thanks{The work of Kuznetsov was supported within the framework of HSE University Basic Research Program and by the the Theoretical Physics and Mathematics Advancement Foundation ``BASIS.''}\institute{Steklov Mathematical Institute of RAS, Russia \\ HSE University, Russia}\email{stephan.kuznetsov@gmail.com}
    \and
    Elaine Pimentel\thanks{Pimentel has received funding from the European Union's Horizon 2020 research and innovation programme under the Marie Sk\l odowska-Curie grant agreement Number 101007627.}
    \institute{University College London, UK}
    \email{e.pimentel@ucl.ac.uk}
    \and Andre Scedrov\institute{University of Pennsylvania, USA}
    \email{scedrov@math.upenn.edu}
    }

\def\authorrunning{Blaisdell, Kanovich, Kuznetsov, Pimentel, Scedrov}

\def\copyrightholders{Blaisdell, Kanovich,\\
 Kuznetsov, Pimentel, Scedrov}

\begin{document}
\maketitle

\begin{abstract} 
In~\cite{DBLP:conf/cade/BlaisdellKKPS22} we introduced a non-associative non-commutative linear logic extended by multimodalities, called subexponentials, licensing local application of structural rules. Here, we further explore this system, considering its classical one-sided multi-succedent classical version, following the exponential-free calculi of~\cite{DBLP:conf/lacl/Buszkowski16} and~\cite{DBLP:journals/sLogica/GrooteL02}, where  the intuitionistic calculus is shown to embed faithfully into the classical fragment. 
\end{abstract}

\section{Introduction}

Following the work of Ajdukiewicz~\cite{Ajdukiewicz} and Bar-Hillel~\cite{Hillel} on categorial grammars, Lambek introduced two versions of non-commutative logical calculi intended to capture grammaticality in natural languages: the associative version~\cite{Lambek1958} and the non-associative version~\cite{Lambek1961OnTC}. In~\cite{Abrusci90} Abrusci  showed  that the associative Lambek calculus corresponds to a non-commutative version of multiplicative intuitionistic linear logic~\cite{Girard:1987uq}.  

Classical linear logic~\cite{Girard:1987uq} is a resource conscious logic, in the sense that
formulae are consumed when used during proofs, unless marked with the modalities $\quest$ and $\bang$ (called  {\em exponentials}).  Formulae marked with such exponentials behave \emph{classically}, so that classical and intuitionistic logics' behaviours can be captured in linear logic.
As it turns out, exponentials are not canonical, in the sense that even having the same scheme for introduction rules, marking the exponentials with different labels (\eg\ $\nbang{i},\nquest{i}$ for $i$ in a set of {\em labels}) does not preserve equivalence, that is, 
$\nbang{i}F\not\equiv\nbang{j}F$ and $\nquest{i}F\not\equiv\nquest{j}F$ if $i\not=j$. This allows for the introduction of a (possibly infinite) set of connectives, called {\em subexponentials}.

Extensions of linear logic/Lambek calculi with subexponentials are considered in~\cite{DBLP:conf/cade/KanovichKNS18,DBLP:journals/mscs/KanovichKNS19,DBLP:conf/cade/BlaisdellKKPS22}. In~\cite{DBLP:conf/cade/KanovichKNS18} a cut-admissible framework for subexponentials in a non-commutative {\em intuitionistic} linear logic was introduced. A  {\em classical} version of this system was presented in~\cite{DBLP:journals/mscs/KanovichKNS19}, showing that, via an appropriate embedding, one can view the `classical' system as a conservative extension of the `intuitionistic' system.
 
In~\cite{DBLP:conf/cade/BlaisdellKKPS22} we extended the work in~\cite{DBLP:conf/cade/KanovichKNS18} by proposing $\acLL$, a non-associative analogue of the previous system. In the present work, we introduce the cut-admissible calculus $\CacLL$ for classical non-associative non-commutative multi-modal linear logic
and show that this is conservative over $\acLL$. 
%

It should be noted that such conservativity results are quite unusual, as they do 
not hold for richer logics which enjoy more structural rules for arbitrary formulae. For example, while classical logic can be adequately represented in the intuitionistic logic via \eg\ G\"oedel-Gentzen double-negation translation~\cite{Avigad1998-FEFOFD,DBLP:journals/corr/abs-1101-5442}, it is easy to see that the other direction has no truth-preserving propositional encodings. In fact, if there were a faithful translation from (propositional) intuitionistic logic into (propositional) classical logic, there would exist a finite matricial decision procedure (\eg\ a truth table) for propositional intuitionistic logic.
The crucial difference between classical logic and substructural ones is that
derivability, in a classical system, of intuitionistically invalid principles requires, besides {\em tertium non datur}, also structural rules. Such principles include, \eg, Peirce's law or Grishin axiom ($A \to (B \vee C) \Rightarrow (A \to B) \vee C$). 

In the substructural setting, the discussion on the conservativity of `classical' systems over `intuitionistic' systems dates back to Schellinx's observation of the analogous result for linear logic with an appropriate choice of connectives~\cite{DBLP:journals/logcom/Schellinx91} (see also~\cite{DBLP:conf/lics/Laurent18}). More specifically, Schellinx proved that fragments of classical linear logic in the language of intuitionistic linear logic are conservative if and only if they do not include the constant $\zero$, or do not include the linear implication. 

This asymmetry is broken if (full powered) subexponentials are added to the system: in~\cite{DBLP:conf/csl/Chaudhuri10} Chaudhuri showed that the conservativity result holds for linear logic with subexponentials. 

In systems not considering the additive constant $\zero$, conservativity is maintained in the absence of other substructural features. For example, in associative non-commutative setting (without subexponentials) an embedding of a `classical' substructural system over an `intuitionistic' one was discussed in~\cite{Pentus1998}. The same holds if subexponentials are added to the scene (see~\cite{DBLP:journals/mscs/KanovichKNS19}). Finally, in the non-associative non-commutative setting, De Groote and Lamarche showed the analogous  
result with exactly the binary multiplicative connectives~\cite{DBLP:journals/sLogica/GrooteL02}.

The present work combines all the aforementioned results, proving that conservativity of `classical'  over the `intuitionistic' holds in the substructural non-associative non-commutative multimodal framework, when the additive constant $\zero$ is not present.
This is a relevant outcome, since conservativity results allow transferring linguistic applications of the intutionistic system to the classical one: both positive (derivability for correct sentences) and negative (non-derivability for incorrect ones) information is preserved. Moreover, if the types for words in a categorial grammar are formulae in the intuitionistic system, then the intuitionistic system is sufficient for all derivations we might need. Therefore, possible {\em new} applications of the classical system in linguistics (that is, applications for which the intuitionistic system is insufficient) would necessarily require using of essentially `classical' types, \ie, formulae which are not translations of intuitionistic ones. We leave the search for such possible applications for further research.

The motivation of the present paper is in the line of De Groote and Lamarche~\cite{DBLP:journals/sLogica/GrooteL02}. As noticed above, in the pure substructural setting the classical system is richer than the intuitonistic one. Symmetries latent in the intuitionistic presentation are now fully observed. Thus, considering a classical counterpart of the intuitionistic system $\acLL$ becomes a theoretical requirement.

The calculi $\acLL$ and $\CacLL$, being non-associative systems, require quite sophisticated structure in sequents. Namely, sequents involve not just sets, multisets, or sequences, but tree-like structures of formulae. This shows some connection to Display Logic~\cite{DBLP:journals/jphil/Belnap82a,DBLP:conf/csl/CloustonDGT13}, which also describe various substructural logics via complex structures of sequents and `display postulates' for different variations of associativity and commutativity. We plan to investigate these connections in future work. 

The rest of the paper is organized as follows. Section~\ref{sec:CaLL} brings a short introduction to linear logic and subexponentials, presents the system $\CacLL$ for classical non-associative non-commutative multi-modal linear logic together with a proof of cut-admissibility of the system. Section~\ref{sec:emb} presents the embedding of $\acLL$ into 
$\CacLL$, showing that one can view the `classical' system as a conservative extension of the `intuitionistic' system. It is also shown that, as in~\cite{DBLP:journals/logcom/Schellinx91,DBLP:journals/mscs/KanovichKNS19}, adding the zero constant is enough for destroying the conservativity. Section~\ref{sec:conc} concludes the paper by pointing some future directions, including a discussion about a focused system. Indeed, since the structural rules are circular, $\CacLL$ is not adequate for proof search. But they can be ``tamed" by eliminating the application of structural rules over structures, and 
restricting the application of the structural subexponential rules to {\em neutral formulae}, in the same way as done in~\cite{DBLP:conf/fossacs/GheorghiuM21} for contraction.

\section{The Classical System}\label{sec:CaLL}
Classical linear logic ($\LL$~\cite{Girard:1987uq}) is a resource conscious logic, in the sense that
formulae are consumed when used during proofs,
unless 
marked with the exponential $\quest$ (whose dual is $\bang$).  Formulae marked with $\quest$ behave \emph{classically}, \ie, they can be contracted 
and weakened 
during proofs.   Propositional $\LL$ connectives include the
additive  conjunction $\with$ and disjunction $\oplus$ and their multiplicative 
versions $\otimes$ and $\lpar$, together with their units: 

\par\nobreak\bgroup%
\begin{tikzpicture}[node distance=1ex]
  \node [matrix of math nodes] (gr) {
    F, G, \dotsc & ::= &
    \node(a){A}; & \mid & \node(tens){F \otimes G}; & \mid & \mathsf{1} & \mid &
    \node(plus){F \oplus G}; & \mid & \mathsf{0} &
    \mid & \node(bang){\mathsf{!} F}; \\
    & \node[right] {\mid} ; &
    \node(a'){ A^\perp}; & \mid &
    \node(par){F \bindnasrepma G}; & \mid & \node(bot){\bot}; & \mid &
    \node(with){F \binampersand G}; & \mid & \node(top){\top}; &
     \mid & \node(qm){\mathsf{?} F}; \\
  } ;
  \node at ($(bang.north east)!.5!(qm.south east)+(1.8,0)$) {
    \refstepcounter{equation}
    \label{eq:gram}
  } ;
  \begin{scope}[on background layer]
    \fill[rounded corners,color=green!5!white]
       ($(a.north west)-(.2,0)$) rectangle ($(a'.south east)-(0,.5)$) ;
    \node at ($(a'.south west)!.5!(a'.south east)-(0.05,.2)$){
      \tiny\scshape literals
    } ;
    \fill[rounded corners,color=blue!5!white]
       (tens.north west) rectangle ($(bot.south east)-(0,.5)$) ;
    \node at ($(par.south west)!.5!(bot.south east)-(0,.2)$) {
      \tiny\scshape multiplicatives
    } ;
    \fill [rounded corners,color=red!5!white]
       (plus.north west) rectangle ($(top.south east)-(0,.5)$) ;
    \node at ($(with.south west)!.5!(top.south east)-(0,.2)$) {
      \tiny\scshape additives
    } ;
    \fill [rounded corners,color=cyan!10!white]
       (bang.north west) rectangle ($(qm.south east)+(.2,-.5)$) ;
    \node[align=flush left] at ($(qm.south west)!.5!(qm.south east)+(.15,-.2)$) {
      \tiny\scshape exp.
    } ;
  \end{scope}
\end{tikzpicture}
\egroup\par\nobreak\noindent
\setcounter{equation}{0}
Note that $(\cdot)^\perp$ (negation) has atomic scope. For an arbitrary formula $F$, $F^\perp$ denotes the result of moving negation inward until it has atomic scope. 
We shall refer to atomic ($A$) and negated atomic  ($A^\bot$) formulae as {\em literals}. 
The connectives in the first line denote the de Morgan dual of the connectives in the second line. Hence,   for atoms $A,B$, the expression $(\bot \with (A \otimes (! B)))^\perp$ denotes
$\one \oplus (A^\perp \lpar (\quest B^\perp))$. 

As is usual for non-associative systems, we consider binary trees of formulae.  To `classicalize', we choose a one-sided sequent and a singular involutive `tight' negation.

\begin{definition}[Structured sequents]
{\em Structures}  include the empty structure $\varnothing$, formulae, or pairs containing structures:
\[
\begin{array}{lcl}
\Gamma & ::= & \varnothing \mid {} F \mid (\Gamma,\Gamma).
\end{array}
\]

Structures are considered up to the following equivalences, which wipe out empty substructures inside a bigger structure: $(\varnothing,\Gamma)$ and $(\Gamma,\varnothing)$ are the same as $\Gamma$.
Thus, any non-empty structure may be regarded as a rooted binary tree whose leaves are labelled with formulae.

%

A \emph{context with several holes}, $\Rx{}\ldots\Ex{}$ is obtained from a structure by replacing designated occurrences of formulae with empty placeholders. Given a context with holes, we write $\Rx{\Delta_1}\ldots\Ex{\Delta_n}$ for the structure which is obtained from $\Rx{}\ldots\Ex{}$ by replacing the placeholders with structures $\Delta_1$, \ldots, $\Delta_n$ (in the given order).

Some of the $\Delta_i$ are allowed to be empty. In this case, the corresponding placeholder is just removed: $(\varnothing,\Gamma)$ and $(\Gamma,\varnothing)$ are replaced by $\Gamma$, and this operation is performed recursively.

A {\em structured sequent} (or simply {\em sequent}) has the form $\seq\Gamma$ where $\Gamma$ is a non-empty structure.
\end{definition}

The rules for the structured system for classical non-associative non-commutative linear logic are depicted in Figure~\ref{fig:FNL}. This is an extension of the system presented in~\cite{DBLP:journals/sLogica/GrooteL02} with the additive connectives.

\begin{figure}[t]
{\sc Propositional rules}
\[
\infer[\tensor]{\seq((\Gamma,\Delta),F\tensor G)}{\seq \Gamma,G &
\seq \Delta,F}
\qquad
\infer[\parr]{\seq\Rx{F\parr G}}{\Rx{(F,G)}}
\]
\[
\infer[\oplus_i]{\seq\Rx{F_1\oplus F_2}}{\seq\Rx{F_i}}
\qquad
\infer[\with]{\seq\Rx{F_1\with F_2}}{\seq\Rx{F_1} &
\seq\Rx{F_2}}
\]
\[
\infer[\perp]{\seq\Rx{\perp}}{\seq\Rx{}}
\qquad
\infer[\one]{ \seq \one}{}
\qquad
\infer[\top]{ \seq \Rx{\top}}{}
\]
{\sc Structural Rules}
\[
\infer[\E]{\seq (\Gamma,\Delta)}{\seq (\Delta,\Gamma)}
\qquad
\infer[\A1]{\seq ((\Gamma,\Delta),\Pi)}{\seq (\Gamma,(\Delta,\Pi))}
\qquad
\infer[\A2]{\seq (\Gamma,(\Delta,\Pi))}{\seq ((\Gamma,\Delta),\Pi)}
\]
{\sc Initial and cut rules}
\[
\infer[\init]{\seq (A,A^{\perp})}{}
\qquad
\infer[\cut]{\seq (\Gamma,\Delta)}{\seq (\Gamma,A) & \seq (A^{\perp},\Delta)}
\]
\caption{Structured system for classical non-associative non-commutative linear logic ($\CNL$).}\label{fig:FNL}
\end{figure}

The structural rules need some clarification.
At the first glance, they look like the rules of commutativity (exchange) and associativity, which could have ruined the whole idea of building a non-associative non-commutative logic. It is important to notice, however, that these rules allow exchange and associativity only {\em on the top level:} e.g., one cannot obtain $\seq ((\Gamma, \Delta), \Psi)$ from $\seq ((\Delta, \Gamma), \Psi)$. This means that the structural rules are a non-associative analogue of cyclic shifts (as in cyclic linear logic~\cite{Yetter}). The level of structural flexibility provided by these rules is discussed below in Section~\ref{sect:equivalence}.

Similar to modal connectives, the exponentials $\bang,\quest$ in $\LL$ are not {\em canonical}~\cite{danos93kgc}, in the sense that if  $i\not= j$ then
$\nbang{i}F\not\equiv\nbang{j}F$ and $\nquest{i}F\not\equiv\nquest{j}F$.
Intuitively, this means that we can mark the exponential with {\em labels} taken from a set $I$ organized in a pre-order $\preceq$ (\ie, reflexive and transitive), obtaining (possibly infinitely-many) exponentials ($\nbang{i},\nquest{i}$ for $i\in I$).
Also as in multi-modal systems, the pre-order  determines the provability relation: 
for a general formula $F$, $\nbang{b}F$ {\em implies} $\nbang{a}F$ iff $a \preceq b$.

Originally~\cite{nigam10jar}, subexponentials could assume only weakening and contraction axioms:
\[
\C:\;\; \nbang{i} F \limp \nbang{i} F
     \otimes \nbang{i} F \qquad \W:\;\; \nbang{i} F \limp \one  
\]
In~\cite{DBLP:conf/cade/KanovichKNS18,DBLP:journals/mscs/KanovichKNS19}, non-commutative systems allowing commutative subexponentials were presented:
\[ \mathsf{E}:\;\; (\nbang{i} F)\otimes G \equiv G\otimes(\nbang{i} F) 
\]
 In~\cite{DBLP:conf/cade/BlaisdellKKPS22}, we went one step further and presented a non-commutative, non-associative linear logic based system with the possibility of assuming associativity
%
\[ \mathsf{A1}:\;\; \nbang{i} F\otimes(G\otimes H) \equiv (\nbang{i} F\otimes G)\otimes H
     \qquad \mathsf{A2}:\;\; (G\otimes H)\otimes \nbang{i} F \equiv G\otimes (H\otimes\nbang{i} F) 
  \]
as well as commutativity and other structural properties. In this paper, we present the classical version of this system.
 
We start by presenting an adaption of simply dependent multimodal linear logics ($\mathsf{SDML}$) appearing in~\cite{DBLP:conf/lpar/LellmannOP17} to the non-associative/commutative case.
The language of non-commutative $\mathsf{SDML}$ is that of (propositional) linear logic
with subexponentials~\cite{DBLP:journals/mscs/KanovichKNS19}.  

\begin{definition}[SDML]\label{def:sdmls}
Let $\mathcal{A}$ be a set of axioms. A (non-associative non-commutative) {\em simply dependent multimodal logical system} ($\mathsf{SDML}$) is given by a triple $\Sigma=(I,\cless,f)$, where $I$ is a set of indices, $(I,\cless)$ is a pre-order, and $f$ is a mapping from $I$ to $2^{\mathcal{A}}$. 
  
If $\Sigma$ is a $\mathsf{SDML}$, then the \emph{logic described by} $\Sigma$ has the modalities $\nbang{i},\nquest{i}$ for every $i \in I$, with the rules of non-associative non-commutative linear logic, together with rules for the axioms $f(i)$ and the interaction axioms $\nbang{j} A\limp \nbang{i} A$ for every $i,j \in I$ with $i \cless j$.

Finally, every $\SDML$ is assumed to be upwardly closed w.r.t. $\preceq$, that is, if $i\preceq j$ then $f(i)\subseteq f(j)$ for all $i,j\in I$.\footnote{This requirement is needed for proving cut-admissibility of the correspondent sequent systems (see~\cite{danos93kgc}).}
\end{definition}

The structured system $\CacLL$ is determined by the system $\CNL$ and the rules in Figure~\ref{fig:CacLL}. $\CacLL$ is the logic described by the $\SDML$ determined by $\Sigma$, with $\mathcal{A}=\{\C,\W,\A1,\A2,\E\}$ where, in the subexponential rule for $\mathsf{S}\in\mathcal{A}$, the respective $s\in I$ is such that $\mathsf{S}\in f(s)$ (\eg\ the subexponential symbol $e$ indicates that $\E\in f(e)$). We will denote by $\nynot{\mathsf{Ax}}\Delta$ the fact that the structure $\Delta$ contains only formulae with top-level  as leaves, each of them assuming the axiom $\mathsf{Ax}$.

As an economic notation, we will write $\upset{i}$ for the \emph{upset} of
the index $i$, \ie, the set $\{j \in I : i \cless j\}$.  We extend this notation to structures in the following way. Let $\Gamma$ be a structure containing only question-marked formulae as leaves.
 If such formulae admit the multiset partition 
 \[
 \{ \nquest{j}F\in\Gamma: i \cless j\}\cup \{ \nquest{k}F \in\Gamma: i \not\cless k\mbox{ and }\W\in f(k)\}
 \]
then $\Gamma^{\upset{i}}$ is the structure obtained from $\Gamma$ by erasing the formulae in the second component of the partition (equivalently, the substructure of $\Gamma$ formed with all and only formulae of the first component of the partition). 
 Otherwise, $\Gamma^{\upset{i}}$ is undefined.

\begin{example}\label{ex:preceq}
Let $\Gamma=(\nquest{i}A,(\nquest{j}B,\nquest{k}C))$ be represented below left, $i\preceq j$ but $i\not\preceq k$, and $\W\in f(k)$. Then $\Gamma^{\upset{i}}=(\nquest{i}A,\nquest{j}B)$ is depicted below right
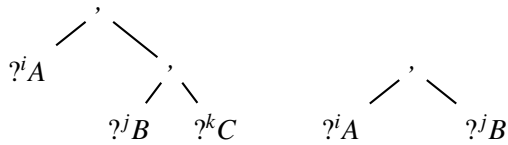
\begin{center}
\begin{tikzpicture}[level distance=2em, 
			level 1/.style={sibling distance=5em},
			level 2/.style={sibling distance=3em},
			every node/.style = {align=center}]]
			\node {,}
			child[thick] { node {$\nquest{i}A$}}
			child[thick] { node {,}
				child[thick] { node {$\nquest{j}B$}}
				child[thick] {node {$\nquest{k}C$}}};
		\end{tikzpicture}
 \qquad
\begin{tikzpicture}[level distance=2em, 
			level 1/.style={sibling distance=5em},
			level 2/.style={sibling distance=3em},
			every node/.style = {align=center}]]
			\node {,}
			child[thick]{node {$\nquest{i}A$}}
			child[thick] { node {$\nquest{j}B$}};
		\end{tikzpicture}
\end{center}
Observe that, if $\W\notin f(k)$, then $\Gamma^{\upset{i}}$ cannot be built. In this case, any derivation of  $\seq(\Gamma,\nbang{i}C)$ cannot start with an application of the promotion rule, similarly to how promotion in $\LL$ cannot be applied in the presence of non-classical contexts. 
\end{example}

\begin{figure}[t]
{\sc Subexponential rules}
\[\infer[\prom]{\seq (\Gamma,\nbang{i}F)}{\seq(\Gamma^{\upset{i}},F)}
\qquad
\infer[\der]{\seq\Rx{\nynot{i}F}}{\seq \Rx{F}}
\]
{\sc Subexponential Structural rules}
\[
\infer[\nynot{}\mathsf{A}1]{\seq((\Delta_1,(\Delta_2,\Delta_3)),\nynot{a1}\Gamma)}{
    \seq(((\Delta_1,\Delta_2),\Delta_3),\nynot{a1}\Gamma)
}
\qquad
\infer[\nynot{}\mathsf{A}2]{\seq(((\Delta_1,\Delta_2),\Delta_3),\nynot{a2}\Gamma)}{
    \seq((\Delta_1,(\Delta_2,\Delta_3)),\nynot{a2}\Gamma)
}
\qquad
\infer[\nynot{}\E]{\seq((\Delta_1,\Delta_2),\nynot{e}\Gamma)}{
\seq((\Delta_2,\Delta_1),\nynot{e}\Gamma)
}
\]
\[
\infer[\nynot{}\W]{\seq\Rx{\nynot{w}\Delta}}{\seq\Rx{}}
\qquad
\infer[\nynot{}\C]{\seq \Rx{}\ldots\Ex{\nquest{c}\Delta}\ldots\Ex{}}{\seq \Rx{\nquest{c}\Delta}\ldots\Ex{\nquest{c}\Delta}}
\]
\caption{Structured system $\CacLL$ for the logic described by $\Sigma$.}\label{fig:CacLL}
\end{figure}

Notice that $\Gamma$, in general, cannot be uniquely computed from $\Gamma^{\upset{i}}$. Indeed, when going from $\Gamma^{\upset{i}}$ to $\Gamma$, one may non-deterministically add formulae of the form ${?}^k C$, where $\W \in f(k)$. Thus, an application of the $\prom$ rule may implicitly contain several applications of the weakening rule $\nynot{}\W$.

\subsection{Structural Equivalence}\label{sect:equivalence}
\newcommand{\desi}[1]{\reallywidehat{#1}}

The structural rules of $\CacLL$ make structures very flexible. 
This flexibility may be pinned by choosing a more sophisticated structures for sequents than rooted binary trees. This corresponds to the move from linearly to circularly ordered sequences of formulae, which is used in the associative case~\cite{Yetter}. In the non-associative situation, this transformation is a bit trickier. Namely, following~\cite{DBLP:journals/sLogica/GrooteL02}, we may represent (non-empty) structures as acyclic graphs (unrooted trees) where each vertex has degree~1 or~3. Vertices of degree~1 are leaves, and they are labelled by individual formulae. Vertices of degree~3 are inner nodes. For each inner node, the cyclic order of its neighbors is maintained.
%
Such graphs with cyclic order on inner nodes are called {\em unrooted cyclically-ordered-neigbor 3-regular trees with leaves}~\cite{DBLP:journals/sLogica/GrooteL02}.


There are many structures equivalent to any given structure, but if we designate a subtree to appear in a particular position then there is a unique representation. This can intuitively described as a keychain with many layers: selecting one for opening a door changes the arrangement of the keys, but not their position in the keychain. 
In this paper, we choose to present all proofs in terms of sequents, and we choose the first right branch of the structure as this designated position.
This is equivalent to flipping the structure around to put a formula into spot, in the same way we would select a key in a keychain.

This works due to the following definitions and technical lemmas (the proofs are in the accompanying technical report~\cite{blaisdell2023explorations}).

\begin{definition}
We define $\sim$, called \emph{structural equivalence}, between two structures to be the reflexive, symmetric, transitive closure of

\begin{align*}
    (\Gamma,\Delta) &\sim (\Delta,\Gamma) \\
    (\Gamma,(\Delta,\Pi)) &\sim ((\Gamma,\Delta),\Pi)
\end{align*}
That is, $\Theta\sim \Xi$ if and only if $\Xi$ is achievable from $\Theta$ using only the structural rules $(\E),(\A1),$ and $(\A2)$.
%
%
\end{definition}

%
%

Note that for the tree representation in~\cite{DBLP:journals/sLogica/GrooteL02}, this corresponds to choosing a particular edge in the graph.
We formalize this by defining the following.

\begin{definition}
For a context with a hole $\Rx{}$, we define the \emph{designated} structure $\desi{\Rx{*}}$ inductively by the following:

\begin{align*}
\desi{(\Rx{*},\Delta)} &:\equiv \desi{(\Delta,\Rx{*})} \\
\desi{(\Gamma,(\Delta,\Px{*}))} &:\equiv \desi{((\Gamma,\Delta),\Px{*})} \\
\desi{(\Gamma,(\Dx{*},\Pi))} &:\equiv \desi{((\Pi,\Gamma),\Dx{*})} \\
\desi{(\Gamma,*)} &:\equiv \Gamma \\
\end{align*}
To ensure that this definition is complete, we also specify the empty case $\desi{*}:\equiv\cdot$.
We call the overline in the notation $\desi{\Rx{*}}$ the \em{designator}.
\end{definition}
Observe that the designator is well defined.
Indeed, first note that the left hand sides are cumulatively exhaustive.  If $*$ is on the left, it is handled by the first case.  If it is on the right, it is either immediately to the right or it is on the right side's left or right branch.  These are handled by the fourth, second, and third cases respectively.
Note also that the left hand sides are mutually exclusive.
Finally, note that recursive application of this definition terminates, because each case places $*$ on the right branch and the depth of $*$ on the right decreases in every subsequent case.

The following lemma shows that this definition indeed gives us an equivalent structure that puts the designated subtree on the right.
\begin{lemma}[Correctness of Designator]
For any structure $\Tx{\Xi}$ with distinguished subtree, we have
\[
(\desi{\Tx{*}},\Xi)\sim\Tx{\Xi}.
\]
\end{lemma}
%
%
%
%
%
%
%
%
%

While the above lemma says that we can designate any substructure as the one that should appear on the right, the following says that this happens uniquely.

\begin{lemma}
If $\Tx{*}\sim\Xx{*}$, then $\desi{\Tx{*}}\equiv\desi{\Xx{*}}$.
\end{lemma}
%
%
%
%
%
%
%
%
%

\begin{corollary}[Uniqueness]
If $(\Gamma,\Pi)\sim(\Delta,\Pi)$ both contain a distinguished occurrence of $\Pi$, then $\Gamma\equiv\Delta$ (and thus $(\Gamma,\Pi)\equiv(\Delta,\Pi)$ as well).
\end{corollary}
%
%

Finally we present another helpful technical lemma which's proof is a straightforward induction on the definition of the designator.
\begin{lemma}[Independent Substructure Preservation]
If $\Delta$ is a substructure of $\Rx{\Delta}\{*\}$ (that does not contain $*$), then $\Delta$ is a substructure of $\desi{\Rx{\Delta}\{*\}}$, and further replacing $\Delta$ by $\Pi$ in $\desi{\Rx{\Delta}\{*\}}$ yields $\desi{\Rx{\Pi}\{*\}}$.
\end{lemma}

Finally, from now on derivations will be considered {\em modulo} designator, in the sense that the operations for determining the designator are not really performed: they should be seen simply as a handy representation of formulae, and not as syntactic manipulations over them. Under this view, we write
\[
\infer=[\sim]{\seq \Gamma}{\seq \Delta}
\quad
\text{for}
\quad
\Gamma\sim\Delta.
\]
only as a convenient representation, not a formal inference rule. 

\newcommand{\cc}{\#_{\C}} 
\newcommand{\rc}{\#_{\R}} 

\subsection{Cut Elimination}
We end this section by presenting the sketch of the proof of admissibility of the cut rule in $\CacLL$. The complete proof is in~\cite{blaisdell2023explorations}.

\begin{theorem}
If a sequent $\seq\Gamma$ is provable in $\CacLL$, then there is a proof in which the $\cut$ rule is not applied.
\end{theorem}
\begin{proof}
We prove cut elimination in a standard syntactic way.  Following \eg\ \cite{DBLP:journals/mscs/KanovichKNS19}  we introduce the following $(\mix)$ rule, and simultaneously eliminate $(\cut)$ and $(\mix)$.

\[
\infer[\mix]{\seq(\Gamma,\Delta\Ex{}\cdots\Ex{})}{
    \seq(\Gamma,\nbang{c}A^{\perp}) &
    \seq(\nynot{c}A,\Delta\Ex{\nynot{c}A}\cdots\Ex{\nynot{c}A})
}
\]
Note that $(\mix)$ is equivalent to a $(\cut)$ followed by (possibly several) applications of $(\C)$, which can be applied since we assume that $\C\in f(c)$.

It is sufficient to prove the claim for one application of $(\cut)$ or $(\mix)$, and we prove this restricted claim jointly for $(\cut)$ and $(\mix)$ by nested induction on $\kappa$, first on the complexity of the $(\cut)$ formula, and then on $\delta$ the depth of the $(\cut)$ or $(\mix)$ application.  Here, we include the $\nynot{c}$ in the cut formula complexity in the case of $(\mix)$.

In each considered case we modify the proof to either remove the $(\cut)$ or $(\mix)$, decrease the complexity of the cut formula, or decrease the depth while maintaining the complexity.
\end{proof}

\newcommand{\trans}[1]{\widehat{#1}}
\newcommand{\ntrans}[1]{\widehat{#1}^{\perp}}

\section{Embedding}\label{sec:emb}
Embedding a classical system with involutive negation into its intuitionistic version is often a matter of finding a ``good translation'' (such as Gentzen-G\"odel's double negation~\cite{DBLP:journals/corr/abs-1101-5442}). The other way around may be tricky, though, sometimes even impossible without collapsing provability to the target logic.

In this section, we will show an embedding 
of the intuitionistic system $\acLL$\footnote{Please refer to~\cite{DBLP:conf/cade/BlaisdellKKPS22} for the rules of the sequent system $\acLL$. In a nutshell,  $\acLL$ is a two-sided version of $\CacLL$, with sequents containing structures as antecedent and a single formula in the succedent. Moreover, the connectives are restricted to: $\with,\oplus, \otimes, \to,\la, \top, \nbang{i},\one$, where $\to,\la$ are the non-commutative linear implications.} into the classical system $\CacLL$, with the same $\SDML$ signature.

Consider the translation \;$\trans{\cdot}$\; on formulae defined below.
\[
\begin{array}{rclcrcl}
    \trans{p} & :\equiv & p & \hspace{0.75in} &
    \trans{A\tensor B} & :\equiv & \trans{A}\tensor\trans{B} \\
    \trans{A\to B} & :\equiv & \trans{A}^{\perp}\parr\trans{B} &&
    \trans{B\la A} & :\equiv & \trans{B}\parr\trans{A}^{\perp} \\
    \trans{A\with B} & :\equiv & \trans{A}\with\trans{B} &&
    \trans{A\oplus B} & :\equiv & \trans{A}\oplus\trans{B} \\
    \trans{\nbang{i}A} & :\equiv & \nbang{i}\trans{A} &&
    \trans{\one} & :\equiv & \one\\
    \trans{\top} & :\equiv & \top
\end{array}
\]
This translation is extended this to structures by the following:
\[
\ntrans{(\Gamma,\Delta)}:\equiv (\ntrans{\Delta},\ntrans{\Gamma})
\]

Note both that the order is reversed by the tight negation and also that we will only every need the negative translation for structures.

We will show this embedding is faithful if no subexponentials license associativity.

\begin{theorem}
If for all labels $i$ in the signature $\Sigma$ we have $f(i)\subseteq \{\C,\W,\E\}$, then an $\acLL$ sequent $\Gamma\seq A$ is provable iff $\seq (\ntrans{\Gamma},\trans{A})$ is provable in $\CacLL$.
\end{theorem}

We start with the easier direction, showing that this embedding is sound.  The embedding is sound, even with the inclusion of associativity.

\begin{lemma}[Soundness]
If an $\acLL$ sequent $\Gamma\seq A$ is provable, then $\seq (\ntrans{\Gamma},\trans{A})$ is provable in $\CacLL$.
\end{lemma}
\begin{proof}
We prove this directly be induction on proofs by showing that the translations of each $\acLL$ rule is a valid $\CacLL$ partial proof.  Consider the bottom rule of a proof. We will show some key cases, please refer to~\cite{blaisdell2023explorations} for the full proof.
%
%
%
%
\[
\infer[\la L]
{\Rx{(B\la A,\Delta)}\seq C}
{\Delta\seq A & \Rx{B}\seq C}
\qquad
\rwto
\qquad
\infer=[\sim]{\seq (\Rxn{(\ntrans{\Delta},\trans{A}\tensor\ntrans{B})},C)}{
\infer[\A1]{\seq (\desi{(\Rxn{*},C)},(\ntrans{\Delta},\trans{A}\tensor\ntrans{B}))}{
\infer[\tensor]{\seq ((\desi{(\Rxn{*},C)},\ntrans{\Delta}),\trans{A}\tensor\ntrans{B})}{
    \seq(\ntrans{\Delta},\trans{A}) &
    \infer=[\sim]{\seq(\desi{(\Rxn{*},C)},\ntrans{B})}{
    \seq(\Rxn{\ntrans{B}},C)
    }
}}}
\]
\[
\infer[\E1]{\Rx{(\nbang{e}\Delta,\Pi)}\seq C}{
\Rx{(\Pi,\nbang{e}\Delta)}\seq C
}
\qquad
\rwto
\qquad
\infer[\sim]{\seq(\Rxn{(\ntrans{\Pi},\nynot{e}\ntrans{\Delta})},\trans{C})}{
\infer[\A2]{\seq(\desi{(\Rxn{*},\trans{C})},(\ntrans{\Pi},\nynot{e}\ntrans{\Delta}))}{
\infer[\nynot{}\E]{\seq((\desi{(\Rxn{*},\trans{C})},\ntrans{\Pi}),\nynot{e}\ntrans{\Delta})}{
\infer=[\A1,\E,\A1]{\seq((\ntrans{\Pi},\desi{(\Rxn{*},\trans{C})}),\nynot{e}\ntrans{\Delta})}{
\infer=[\sim]{\seq(\desi{(\Rxn{*},\trans{C})},(\nynot{e}\ntrans{\Delta},\ntrans{\Pi}))}{
\seq(\Rxn{(\nynot{e}\ntrans{\Delta},\ntrans{\Pi})},\trans{C})
}}}}}
\]
\end{proof}
We now prove the more surprising direction, that the embedding of the intuitionistic system into the classical system is complete.

We start by 
proposing a counter on formulae, which is an extension of the counter defined in~\cite{DBLP:conf/rta/KanovichKMS17}, in its turn an extension of the counter in~\cite{Pentus1998}.

\newcommand{\ctr}{\natural}

\begin{definition}
For a $\CacLL$ formula $A$, we define the integer number $\ctr(A)$ by induction as follows.
\[
\begin{array}{rcl c rcl}
    \ctr(p) & := &  0 & \hspace{0.5in} &
    \ctr(A\parr B) & := & \ctr(A) + \ctr(B) - 1 \\
    \ctr(\bar{p}) & := & 1 &&
    \ctr(A\tensor B) & := & \ctr(A)+\ctr(B) \\
    \ctr(\one) & := & 0 &&
    \ctr(A\oplus B) = \ctr(A\with B) & := & \ctr(A) \\
    \ctr(\perp) & := & 1 &&
    \ctr(\nynot{i}A) = \ctr(\nbang{i}A) & := & \ctr(A)
\end{array}
\]
and extend to structures by
\[
\ctr((\Gamma,\Delta)):=\ctr(\Gamma)+\ctr(\Delta)
\]
\end{definition}

We need this counter for the following technical lemmas, which are easily proven by  straightforward induction.

\begin{lemma}
For any $\acLL$ formula $C$, we have $\ctr(\trans{C})=0$ and $\ctr(\ntrans{C})=1$.
\end{lemma}


\begin{corollary}
A formula cannot be both of the form $\trans{A}$ and $\ntrans{B}$.
\end{corollary}

\begin{lemma}
Let $\seq\Gamma$ be a provable sequent with $n$ formulae where every formula is of the form $\trans{C}$ or $\ntrans{C}$.  Then $\sum_{A\in\Gamma}\ctr(A)=n-1$.
\end{lemma}


\begin{definition}
We say that a sequent with exactly one formula of the form $\trans{C}$ and the rest of the form $\ntrans{C}$ is \emph{intuitionistically polarizable}.  We call the formula of the form $\trans{C}$ the \em{positive formula}.
\end{definition}

\begin{lemma}[Intuitionistic Polarization]
If a sequent with all formulae are of the form $\trans{C}$ or $\ntrans{C}$ is provable, then exactly one of the formulae is of the form $\trans{C}$, \ie\ it is intuitionistically polarizable.
\end{lemma}
\begin{proof}
Let $n$ be the number of formulae in $\Gamma$.  Since by previous lemmas $\ctr(\trans{C})=0$, $\ctr(\ntrans{C})=1$, and $\sum_{A\in\Gamma}\ctr(A)=n-1$, there are exactly $n-1$ formulae of the form $\ntrans{C}$.
\end{proof}

\begin{remark}
Note that any intuitionistically polarizable sequent is structurally equivalent to a unique sequent of the form $(\ntrans{\Gamma},\trans{C})$.
\end{remark}

We now sketch the proof of completeness. The full proof can be found in~\cite{blaisdell2023explorations}.
\begin{lemma}[Completeness]
Let $\Sigma$ be a subexponential signature where all labels $i$ have $f(i)\subseteq \{\C,\W,\E\}$ and let $\Gamma\seq A$ be an $\acLL$ sequent.  If $\seq (\ntrans{\Gamma},\trans{A})$ is provable in $\CacLL$, then $\Gamma\seq A$ is provable in $\acLL$.
\end{lemma}

\begin{proof}
We prove the theorem by induction on the length of $\CacLL$ proofs.  
If the first nonstructural rule is $(\tensor)$, we must consider the following subcases.
\[
\infer=[\sim]{\seq ((\ntrans{\Delta},\ntrans{\Gamma}),\trans{A}\otimes\trans{B})}{
\infer[\otimes]{\seq ((\ntrans{\Delta},\ntrans{\Gamma}),\trans{A}\otimes\trans{B})}{
\seq (\ntrans{\Gamma},\trans{A}) &
\seq (\ntrans{\Delta},\trans{B})
}}
\quad
\rwto
\quad
\infer[\otimes R]{\Gamma,\Delta\seq A\tensor B}{
\Gamma\seq A &
\Delta\seq B
}
\]
\[
\infer=[\sim]{\seq(\Rxn{\ntrans{\Delta},\trans{A}\otimes\ntrans{B}},\trans{C})}{
\infer[\otimes]{\seq((\desi{(\Rxn{*},\trans{C})},\ntrans{\Delta}),\trans{A}\otimes\ntrans{B})}{
\seq(\ntrans{\Delta},\trans{A}) &
\deduce{\seq(\desi{(\Rxn{*},\trans{C})},\ntrans{B})}{
    \seq (\Rxn{\ntrans{B}},\trans{C}) \sim
}
}}
\quad
\rwto
\quad
\infer[\la L]{\Rx{B\la A,\Delta}\seq C}{
    \Delta\seq A &
    \Rx{B}\seq C
}
\]
\[
\infer=[\sim]{\seq(\Rxn{\ntrans{B}\otimes\trans{A},\ntrans{\Delta}},\trans{C})}{
\infer[\otimes]{\seq((\ntrans{\Delta},\desi{(\Rxn{*},\trans{C})}),\ntrans{B}\otimes\trans{A})}{
\deduce{\seq(\desi{(\Rxn{*},\trans{C})},\ntrans{B})}{
    \seq (\Rxn{\ntrans{B}},\trans{C}) \sim
} &
\seq(\ntrans{\Delta},\trans{A})
}}
\quad
\rwto
\quad
\infer[\la L]{\Rx{\Delta,A\to B}\seq C}{
    \Delta\seq A &
    \Rx{B}\seq C
}
\]
There are two remaining ways that $\tensor$ can appear in the translation of a sequent and be principal.
\[
\infer=[\sim]{\seq(\Rxn{\ntrans{\Delta},\ntrans{B}\otimes\trans{A}},\trans{C})}{
\infer[\otimes]{\seq((\desi{(\Rxn{*},\trans{C})},\ntrans{\Delta}),\ntrans{B}\otimes\trans{A})}{
    \seq(\ntrans{\Delta},\ntrans{B}) &
    \seq(\desi{(\Rxn{*},\trans{C})},\trans{A})
}}
\qquad
\infer=[\sim]{\seq(\Rxn{\trans{A}\otimes\ntrans{B},\ntrans{\Delta}},\trans{C})}{
\infer[\otimes]{\seq((\ntrans{\Delta},\desi{(\Rxn{*},\trans{C})}),\trans{A}\otimes\ntrans{B})}{
    \seq(\desi{(\Rxn{*},\trans{C})},\trans{A}) &
    \seq(\ntrans{\Delta},\ntrans{B})
}}
\]
However, the premises are not intuitionistically polarizable, and therefore cannot be provable; in other words, these cases are impossible.

Most interestingly we have subexponentially licensed structural rules.
The positive formula in the sequent cannot be of the form $\nynot{i}A$, and is thus not part of the active substructure of any subexponential structural rules.  Hence, by the independent substructure lemma, we can make the following transformations, where $\desi{\Rx{*}}^r$ indicates reversing $\Rx{*}$, designating, and reversing back.
\[
\infer=[\sim]{\seq(\desi{\Rxn{\nynot{c}\ntrans{\Delta}}\Ex{}\Ex{*}},\trans{C})}{
\infer[\C]{\seq(\Rxn{\nynot{c}\ntrans{\Delta}}\Ex{}\Ex{\trans{C}})}{
\deduce{\seq(\Rxn{\nynot{c}\ntrans{\Delta}}\Ex{\nynot{c}\ntrans{\Delta}}\Ex{\trans{C}})}{
\seq(\desi{\Rxn{\nynot{c}\ntrans{\Delta}}\Ex{\nynot{c}\ntrans{\Delta}}\Ex{*}},\trans{C})\sim
}}}
\quad
\rwto
\quad
\infer[\C]{\desi{\Rx{\nbang{c}\Delta}\Ex{*}}^r\seq C}{
\desi{\Rx{\nbang{c}\Delta}\Ex{\nbang{c}\Delta}\Ex{*}}^r\seq C
}
\]
\end{proof}

We finish this section with two observations regarding some of our choices on rules and notation. First, it should be clear now the necessity of the top level structural rules in the system $\CNL$.  Since we expect completeness over the intuitionistic system, translations of provable sequents should themselves be provable. For example, $A\to B\seq A\to B$ translates corresponds to the one sided sequent $\seq (B^{\perp}\tensor A,A^{\perp}\parr B)$, whose proof requires top level exchange.
\[
\infer[\parr]{\seq (B^{\perp}\tensor A,A^{\perp}\parr B)}{
\infer[\E]{\seq (B^{\perp}\tensor A,(A^{\perp},B))}{
\infer[\tensor]{\seq ((A^{\perp},B),B^{\perp}\tensor A)}{
    \infer[\init]{\seq (A^{\perp},A)}{} &
    \infer[\init]{\seq (B,B^{\perp})}{}
}}}
\]
In non-associative systems, currying requires application of associativity.  Loosely, the deduction theorem is a top level currying, and this is captured by the admission of top level associativity in the classical system.  We see this behavior in the proof of the translation of $B\seq (A\to A\tensor B)$, i.e. $\seq (B^{\perp},A^{\perp}\parr(A\tensor B))$.  This cannot be proven without top level associativity.
\[
\infer[\parr]{\seq (B^{\perp},A^{\perp}\parr(A\tensor B))}{
\infer[\A2]{\seq (B^{\perp},(A^{\perp},A\tensor B))}{
\infer[\tensor]{\seq ((B^{\perp},A^{\perp}),A\tensor B)}{
    \infer[\init]{\seq(B^{\perp},B)}{} &
    \infer[\init]{\seq(A^{\perp},A)}{}
}}}
\]

Finally, we would like to note that in our classical system we follow the right-handed presentation $\seq \Gamma$, which is traditional in logic. It is also possible to consider the dual, left-handed presentation    $\Gamma^\perp \Rightarrow$, as in Buszkowski~\cite{DBLP:conf/lacl/Buszkowski16}, which is perhaps closer to the notation in type-logical, formal linguistics. An intuitionistic sequent $\Gamma \seq A$ would be translated into the left-handed classical system as $(A^\perp , \Gamma) \seq$. This translation is also conservative.

\subsection{Incompleteness with Associativity}
Our completeness excludes subexponentials licensing associativity.  To see why, consider the formula
\[
((a\tensor b)\tensor \nbang{a}c)\to(a\tensor (b\tensor \nbang{a}c))
\]
which encodes the converse to the rule $(\A2)$ of $\acLL$.  An exhaustive search finds that there is no cut-free proof of this in $\acLL$, but its translation has the following proof in $\CacLL$.

\[
\infer=[\parr]{\seq(\nynot{a}\bar{c}\parr(\bar{b}\parr\bar{a}))\parr(a\tensor (b\tensor \nbang{a}c)}{
\infer=[\A1,\E]{\seq((\nynot{a}\bar{c},(\bar{b},\bar{a})),a\tensor (b\tensor \nbang{a}c))}{
\infer[\nynot{}\A1]{\seq((\bar{b},\bar{a}),a\tensor (b\tensor \nbang{a}c),\nynot{a}\bar{c})}{
\infer=[\E,\A2,\A2]{\seq((\bar{b},(\bar{a},a\tensor (b\tensor \nbang{a}c))),\nynot{a}\bar{c})}{
\infer[\tensor]{\seq(((\nynot{a}\bar{c},\bar{b}),\bar{a}),a\tensor (b\tensor \nbang{a}c))}{
    \infer[\init]{\seq(\bar{a},a)}{} &
    \infer[\tensor]{\seq((\nynot{a}\bar{c},\bar{b}),b\tensor \nbang{a}c)}{
        \infer[\init]{\seq(\bar{b},b)}{} &
        \infer[\init]{\seq(\nynot{a}\bar{c},\nbang{a}c)}{}
    }
}}}}}
\]

\subsection{Extending $\acLL$}
We can recapture associativity by adding in more rules to the `intuitionistic' system $\acLL$.  While $\acLL$ has two subexponential labels for associativity and two subexponential rules for associativity, we consider an expanded system, still with two labels for associativity, but with six rules.

\begin{definition}
Let $\acLL^+$ be the logic containing all the rules of $\acLL$ except $(\A1)$ and $(\A2)$ with the addition of the following six subexponential associativity rules.

\[
\infer[\mathsf{A1L}]{\Rx{(\nbang{a1}\Delta_1,(\Delta_2,\Delta_3))}\seq G}{\Rx{((\nbang{a1}\Delta_1,\Delta_2),\Delta_3)}\seq G}
\qquad
\infer[\mathsf{A1M}]{\Rx{((\Delta_1,\nbang{a1}\Delta_2),\Delta_3)}\seq G}{\Rx{(\Delta_1,(\nbang{a1}\Delta_2,\Delta_3))}\seq G}
\qquad
\infer[\mathsf{A1R}]{\Rx{(\Delta_1,(\Delta_2,\nbang{a1}\Delta_3))}\seq G}{\Rx{((\Delta_1,\Delta_2),\nbang{a1}\Delta_3)}\seq G}
\]
\[
\infer[\mathsf{A2L}]{\Rx{((\nbang{a2}\Delta_1,\Delta_2),\Delta_3)}\seq G}{\Rx{(\nbang{a2}\Delta_1,(\Delta_2,\Delta_3))}\seq G}
\qquad
\infer[\mathsf{A2M}]{\Rx{(\Delta_1,(\nbang{a2}\Delta_2,\Delta_3))}\seq G}{\Rx{((\Delta_1,\nbang{a2}\Delta_2),\Delta_3)}\seq G}
\qquad
\infer[\mathsf{A2R}]{\Rx{((\Delta_1,\Delta_2),\nbang{a2}\Delta_3)}\seq G}{\Rx{(\Delta_1,(\Delta_2,\nbang{a2}\Delta_3))}\seq G}
\]
\end{definition}

Note that the rules $(\AOL)$ and $(\ATR)$ of $\acLL^+$ are exactly the rules $(\A1)$ and $(\A2)$ of $\acLL$, respectively, so $\acLL^+$ is a stronger system.

The proof of the next theorem is exactly as in the previous completeness theorem, with the addition of a case for the rules $(\A1)$ and $(\A2)$, and it is shown in~\cite{blaisdell2023explorations}.
\begin{theorem}[Completeness with Associativity]
Let $\Gamma\seq A$ be an $\acLL^+$ sequent (whose signature may include $\A1$ and $\A2$).  If $\seq(\ntrans{\Gamma},\trans{A})$ is provable in $\CacLL$, then $\Gamma\seq A$ is provable in $\acLL^+$.
\end{theorem}

\subsection{Incompleteness with Additive Constants}
If we extend $\acLL$ with the $\zero$ constant, governed by the following rule, we lose completeness.

\[
\infer[\zero L]{\Rx{\zero}\seq C}{}
\]

Adapting a counterexample of \cite{DBLP:journals/mscs/KanovichKNS19}, who themselves adapt a counterexample of \cite{DBLP:journals/logcom/Schellinx91}, we consider the following sequent.

\[
\nbang{a}((r\la(0\to q))\la p),(s\la p)\to 0\seq r
\]

By exhaustive proof search, we find that this is not provable in the extended intuitionistic system.  However, the translation, with $\trans{0}:\equiv 0$, has the following proof in $\CacLL$.

\[
\infer[~]{\seq ((\top\tensor(s\parr p^{\perp}),\nynot{a}(p\tensor((\top\parr q)\tensor r^{\perp}))),r)}{
\infer[\tensor]{\seq((\nynot{a}(p\tensor((\top\parr q)\tensor r^{\perp})),r),\top\tensor(s\parr p^{\perp}))}{
    \infer[\top]{\seq\top}{} &
    \infer[\parr]{\seq((\nynot{a}(p\tensor((\top\parr q)\tensor r^{\perp})),r),s\parr p^{\perp})}{
    \infer[\sim]{\seq((\nynot{a}(p\tensor((\top\parr q)\tensor r^{\perp})),r),(s,p^{\perp}))}{
    \infer[\A1]{\seq((r,(s,p^{\perp})),\nynot{a}(p\tensor((\top\parr q)\tensor r^{\perp})))}{
    \infer[\der]{\seq(((r,s),p^{\perp}),\nynot{a}(p\tensor((\top\parr q)\tensor r^{\perp})))}{
    \infer[\tensor]{\seq(((r,s),p^{\perp}),p\tensor((\top\parr q)\tensor r^{\perp}))}{
        \infer[\tensor]{\seq((r,s),(\top\parr q)\tensor r^{\perp})}{
            \infer[\init]{\seq (r,r^{\perp})}{} &
            \infer[\parr]{\seq (s,\top\parr q)}{
                \infer[\top]{\seq (s,(\top,q))}{}
            }
        } &
        \infer[\init]{(p^{\perp},p)}{}
    }}}}}
}
}
\]

\section{Conclusion and future work}\label{sec:conc}
Regarding proof theoretic aspects of $\CacLL$, the circularity of both the structural rules over structures and subexponentials may cause meaningless steps in derivations.
The focusing  discipline~\cite{Andreoli:1992} is
determined by the alternation of {\em focused} and {\em unfocused} phases in the proof construction.
In the unfocused phase, inference rules can be applied eagerly and no
backtracking is necessary;
in the  focused phase, on the other hand, either context restrictions apply, or choices
within inference rules can lead to failures for which one may need to
backtrack. These phases are totally determined by the polarities of formulae:
provability is preserved when applying right/left rules for negative/positive formulae respectively, but not necessarily in other cases. 

In the near future, we  plan to propose a focused system for $\CacLL$. We start by observing that, since derivations will be considered {\em modulo} designator, the (circular) structural rules in Figure~\ref{fig:FNL} can be dropped. The only other point  of circularity comes from the subexponential structural rules. For the remaining linear logic connectives,  polarization and focusing is well understood. 
Weakening can be absorbed into the rules and it seems possible to restrict contraction to ``neutral'' structures such as proposed in~\cite{DBLP:conf/fossacs/GheorghiuM21}, meaning that it can be applied only before focusing on a formula. The only point of atention would be associativity, and this is under investigation at the moment. 

Also, our subexponential extension of the Lambek non-associative calculus $\NL$~\cite{Lambek1961OnTC} makes it possible to use local associativity, controlled by appropriate subexponentials, instead of the global associativity. The usefulness of this more fine-grained control of associativity is best seen in the linguistic examples considered in~\cite{DBLP:conf/cade/BlaisdellKKPS22} that involve both non-associativity and associativity, such as, ``The superhero whom Hawkeye killed was incredible''. 

A dual approach to combining associative and non-associative features is the (associative) Lambek calculus with brackets, developed by Morrill~\cite{Morrill1992,DBLP:journals/jolli/Morrill19} and Moortgat~\cite{Moortgat1996}. In that approach the underlying system is associative and bracket modalities control local non-associativity.  In the setting without subexponentials, Kurtonina~\cite{Kurtonina1995}  showed that the Lambek non-associative calculus $\NL$~\cite{Lambek1961OnTC} can be conservatively embedded in the calculus with brackets, which is a conservative extension of the Lambek associative calculus $\mathsf{L}$~\cite{Lambek1958} by construction. In the presence of subexponentials, the exact relationship between the two approaches remains to be investigated.

We end with some discussion about complexity.
Even when enriched by subexponentials, classical non-associative non-commutative linear logic remains conservative over intuitionistic non-associative non-commutative linear logic, for certain choices of rules or fragments of the language.

Conservativity allows one to port the undecidability result of \cite{DBLP:conf/cade/BlaisdellKKPS22} which adapts a result of Chvalovsk\'{y} \cite{Chvalovsky2015}.  This in particular would give undecidability of $\CacLL$ in full generality.

However, since strictly more is expressible in $\CacLL$ than $\acLL$, decidability results, as in Busz\-kowski's work \cite{DBLP:conf/lacl/Buszkowski16}, would be stronger and would have implications in the reverse direction.

Further, Buszkowski \cite{DBLP:conf/lacl/Buszkowski16} also shows that classical non-associative non-commutative multiplicative linear logic generates context-free grammars as a categorial grammar.  From the intuitionistic direction, subexponentials are useful tools for modeling certain linguistic phenomena, 
including specifically associative subexponentials in non-associative systems \cite{DBLP:conf/cade/BlaisdellKKPS22}.  There is much left to explore regarding the application of subexponentials to linguistic analysis in classical systems.

\noindent
{\bf Acknowledgment.} We would like to thank  Dale Miller for his suggestion that we investigate the classical version of our calculus in~\cite{DBLP:conf/cade/BlaisdellKKPS22}, which motivated the present work.

\bibliography{biblio}

%

\end{document}